%% file: NSW.tex
\documentclass[11pt]{article}

\usepackage{amssymb,amsthm,amsmath,amssymb,wrapfig,dsfont,authblk}
\usepackage{thmtools,thm-restate}
\usepackage[margin=0.9in]{geometry}
\usepackage{booktabs}
\usepackage{cleveref}
\usepackage[linesnumbered,lined,boxed,ruled]{algorithm2e}
\SetKwInOut{Parameter}{Parameter}
\SetKwRepeat{Do}{do}{while}
\SetKwFor{ForEach}{for each}{do}{endfor}

\newtheorem{lemma}{Lemma}

\theoremstyle{definition}
\newtheorem{remark}{Remark}
\newtheorem{example}{Example}

\allowdisplaybreaks

\newcommand{\argmax}[1]{\underset{#1}{\operatorname{arg}\,\operatorname{max}}\;}

\newcommand{\AlgBinary}{\textrm{\textsc{Alg-Binary}}}
\newcommand{\AlgId}{\textrm{\textsc{Alg-Identical}}}

\newcommand{\APXhard}{\text{APX-hard}}

\newcommand{\EF}{\mathrm{EF}}

\newcommand{\EFx}{\mathrm{EFx}}
\newcommand{\F}{{\mathcal F}}

\renewcommand{\H}{\mathcal H}

\newcommand{\I}{{\mathcal I}}

\newcommand{\N}{{\mathbb N}}

\newcommand{\NPhard}{\text{NP-hard}}
\newcommand{\NW}{\mathrm{NSW}}
\newcommand{\NSW}{\mathrm{NSW}}
\renewcommand{\O}{\mathcal O}

\newcommand{\Partition}{\textrm{\textsc{Partition}}}

\newcommand{\PTAS}{\text{PTAS}}

\newcommand{\V}{\mathcal V}

\begin{document}

\title{\bfseries Greedy Algorithms for Maximizing Nash Social Welfare}

\author{Siddharth Barman\thanks{Indian Institute of Science. \texttt{barman@iisc.ac.in}  \\ \hspace*{13pt} Supported in part by a Ramanujan Fellowship (SERB - {SB/S2/RJN-128/2015}).}, \quad Sanath Kumar Krishnamurthy\thanks{Chennai Mathematical Institute. \texttt{sanathkumar9@cmi.ac.in}}, \quad Rohit Vaish\thanks{Indian Institute of Science. \texttt{rohitv@iisc.ac.in}}}

\date{}
\maketitle

\begin{abstract}
We study the problem of fairly allocating a set of indivisible goods among agents with additive valuations. The extent of fairness of an allocation is measured by its Nash social welfare, which is the geometric mean of the valuations of the agents for their bundles. While the problem of maximizing Nash social welfare is known to be $\APXhard{}$ in general, we study the effectiveness of \emph{simple}, \emph{greedy} algorithms in solving this problem in two interesting special cases.

First, we show that a simple, greedy algorithm provides a $1.061$-approximation guarantee when agents have \emph{identical} valuations, even though the problem of maximizing Nash social welfare remains \NPhard{} for this setting. Second, we show that when agents have \emph{binary} valuations over the goods, an \emph{exact} solution (i.e., a Nash optimal allocation) can be found in polynomial time via a greedy algorithm. Our results in the binary setting extend to provide novel, exact algorithms for optimizing Nash social welfare under \emph{concave} valuations. Notably, for the above mentioned scenarios, our techniques provide a \emph{simple} alternative to several of the existing, more sophisticated techniques for this problem such as constructing equilibria of Fisher markets or using real stable polynomials.
\end{abstract}

\section{Introduction}
\label{sec:Introduction}

We study the problem of fairly allocating a set of indivisible goods among agents with additive valuations for the goods. The fairness of an allocation is quantified by its \emph{Nash social welfare} \cite{N50bargaining,KN79nash}, which is the geometric mean of the valuations of the agents under that allocation. The notion of Nash social welfare has traditionally been studied in the economics literature for \emph{divisible} goods \cite{M04fair}, where it is known to possess strong \emph{fairness} and \emph{efficiency} properties \cite{V74equity}. Besides, this notion is also attractive from a \emph{computational} standpoint: For divisible goods, the Nash optimal allocation can be computed in polynomial time using the convex program of Eisenberg and Gale \cite{EG59consensus}.

For \emph{indivisible} goods, Nash social welfare once again provides notable fairness and efficiency guarantees \cite{CKM+16unreasonable}. However, the computational results in this setting are drastically different from its divisible counterpart. Indeed, it is known that the problem of maximizing Nash social welfare for indivisible goods is \APXhard{} when agents have additive valuations for the goods \cite{Lee17APX}. On the algorithmic side, the first constant-factor (specifically, $2.89$) approximation for this problem was provided by Cole and Gkatzelis \cite{CG15approximating}. This approximation factor was subsequently improved to $e$ \cite{AGS+17nash}, $2$ \cite{CDG+17convex} and, most recently, to $1.45$ \cite{BKV17Finding}. Similar approximation guarantees have also been developed for more general market models such as piecewise linear concave (PLC) utilities \cite{AGM+18nash}, budget additive valuations \cite{GHM18Approximating}, and multi-unit markets \cite{BGH+17earning}.
By and large, these approaches rely on either constructing an appropriate equilibrium of a Fisher market and later rounding it to an integral allocation, or using real stable polynomials. Although these approaches offer strong approximation guarantees in very general market models, they are also (justifiably) more involved and often lack a \emph{combinatorial} interpretation. Our interest in this work, therefore, is to understand the power of simple, combinatorial algorithms (in particular, \emph{greedy} techniques) in solving interesting special cases of this problem.

\paragraph{Our results and techniques}
We consider the Nash social welfare objective ($\NSW$) as a measure of fairness in and of itself, and  develop greedy algorithms for maximizing $\NSW{}$ either exactly or approximately. We focus on two special classes of additive valuations, namely \emph{identical} valuations (i.e., for any good $j$ and any pair of agents $i,k$, the value of the good $j$ for $i$ is equal to the value of the good for $k$; $v_{i,j} = v_{k,j}$) and \emph{binary} valuations (i.e., for every agent $i$ and good $j$, agent $i$'s value for $j$ is either $0$ or $1$; $v_{i,j} \in \{0,1\}$). The class of identical valuations is well-studied in the approximation algorithms literature, and binary valuations capture the setting where each agent finds a good either acceptable or not.

For \emph{identical} valuations, we show that a simple greedy algorithm provides a $1.061$-approximation to the optimal Nash social welfare (\Cref{thm:IdenticalVals}). Note that the problem of maximizing Nash social welfare remains $\NPhard{}$ even for identical valuations (via a reduction from the $\Partition{}$ problem \cite{RR10complexity}). Our algorithm (Algorithm~\ref{alg:IdenticalVals}) works by allocating the goods one by one in descending order of value. At each step, a good is allocated to the agent with the least valuation. This implicitly corresponds to greedily choosing an agent that provides the maximum improvement in $\NSW$. We show that the allocation returned by our algorithm satisfies an approximate version of \emph{envy-freeness} property (\Cref{lem:AlgIdentical_EFx}), and that any allocation with this property gives the desired approximation guarantee (\Cref{lem:EFx_implies_approx_Nash}). We remark that a polynomial time approximation scheme ($\PTAS{}$) is already known for this problem \cite{NR14minimizing}. However, this scheme uses Lenstra's algorithm for integer programs \cite{L83integer} as a subroutine, and hence is not combinatorial. Moreover, despite being polynomial time in principle, the actual running time of such algorithms often scales rather poorly. By contrast, our algorithm involves a single sorting step and at most one $\min(\cdot)$ operation for each good, thus it requires $\O(m \log m + mn)$ time overall.

For \emph{binary} valuations, we show that an \emph{exact} solution to the problem (i.e., a Nash optimal allocation) can be found by a greedy algorithm in polynomial time (\Cref{thm:BinaryVals}). However, unlike the algorithm for identical valuations which greedily picks an agent, our algorithm for binary valuations makes a greedy decision with respect to \emph{swaps} (or \emph{chains} of swaps) between a \emph{pair} of agents. A swap refers to taking a good away from one agent and giving it to another agent. A chain of swaps refers to a sequence of agents $u_1,u_2,\dots,u_\ell$ and goods $j_1,j_2,\dots,j_{\ell-1}$ such that $j_1$ is swapped from $u_1$ to $u_2$, $j_2$ is swapped from $u_2$ to $u_3$, and so on. Given any suboptimal allocation, our algorithm (Algorithm~\ref{alg:BinaryVals}) checks for every pair of agents whether there exists a chain of swaps between them that improves $\NW$. The pair that provides the greatest improvement is chosen, and the corresponding swaps made. We show that the algorithm makes substantial progress towards the Nash optimal after each such reallocation (\Cref{lem:Binary_main}), which provides the desired running time and optimality guarantees. An interesting feature of our algorithm is that the guarantee for additive valuations extends to a more general utility model where the valuation of each agent is a \emph{concave} function of its cardinality (i.e., the number of nonzero-valued goods in its bundle).\footnote{Notice that for binary and additive utilities, the valuation of each agent is a \emph{linear} function of its cardinality.} Prior work \cite{DS15maximizing} has shown that a Nash optimal can be found efficiently under binary valuations (via reduction to minimum-cost flow problem). However, these techniques crucially rely on valuations depending \emph{linearly} on cardinality, and it is unclear how to extend these to the aforementioned utility model. Our results, therefore, provide novel, exact algorithms for maximizing Nash social welfare under concave valuations.

\section{Preliminaries}
\label{sec:Preliminaries}

\paragraph{Problem instance} 
An \emph{instance} $\langle [n], [m], \V \rangle$ of the fair division problem is defined by (1) the set of $n \in \N$ \emph{agents} $[n] = \{1,2,\dots,n\}$, (2) the set of $m \in \N$ \emph{goods} $[m] = \{1,2,\dots,m\}$, and (3) the \emph{valuation profile} $\V = \{v_1,v_2,\dots,v_n\}$ that specifies the preferences of each agent $i \in [n]$ over the set of goods $[m]$ via a \emph{valuation function} $v_i: 2^{[m]} \rightarrow \mathbb{Z_+} \cup \{0\}$. Throughout, the valuations are assumed to be \emph{additive}, i.e., for any agent $i \in [n]$ and any set of goods $G \subseteq [m]$, $v_i(G) := \sum_{j \in G} v_i(\{j\})$, where $v_i(\{\emptyset\}) = 0$. We use the shorthand $v_{i,j}$ instead of $v_i(\{j\})$ for a singleton good $j \in [m]$. We use $\Gamma_i$ to denote the set of goods that are positively valued by agent $i$, i.e., $\Gamma_i := \{ j \in [m] \, : \, v_{i,j} > 0 \}$.

\paragraph{Binary and identical valuations}
We say that agents have \emph{binary} valuations if for each agent $i \in [n]$ and each good $j \in [m]$, $v_{i,j} \in \{0,1\}$. In addition, we say that agents have \emph{identical} valuations if for any good $j \in [m]$ and any pair of agents $i,k \in [n]$, we have $v_{i,j} = v_{k,j}$. For identical valuations, we will assume, without loss of generality, that the value of each good is nonzero.

\paragraph{Allocation}
An \emph{allocation} $A \in \{0,1\}^{n \times m}$ refers to an $n$-partition $(A_1,\dots,A_n)$ of $[m]$, where $A_i \subseteq [m]$ is the \emph{bundle} allocated to the agent $i$. Let $\Pi_n([m])$ denote the set of all $n$ partitions of $[m]$. Given an allocation $A$, the valuation of an agent $i \in [n]$ for the bundle $A_i$ is $v_i(A_i) = \sum_{j \in A_i} v_{i,j}$. An allocation is said to be \emph{non-wasteful} if it does not assign a zero-valued good to any agent, i.e., for each agent $i$ and each good $j \in A_i$, we have $v_{i,j} > 0$. We will use the terms \emph{allocation} and \emph{partition} interchangeably whenever the set of goods $[m]$ is clear from the context.

\paragraph{Nash social welfare}
Given an instance $\I = \langle [n], [m], \V \rangle$ and an allocation $A$, the \emph{Nash social welfare} of $A$ is given by $\NW(A) := \left( \prod_{i \in [n]} v_i(A_i) \right)^{1/n}$. An allocation $A^*$ said to be \emph{Nash optimal} if $A^* \in \arg\max_{A \in \Pi_n([m])} \NW(A)$. An allocation $B$ is said to be a \emph{$\beta$-approximation} (where $0 \leq \beta \leq 1$) for the instance $\I$ if $\NW(B) \geq \beta \cdot \NW(A^*)$. For the approximation guarantees to be meaningful, we will assume that the Nash optimal for any given instance has nonzero Nash social welfare.

\section{Main Results}
\label{sec:Results}

We provide two main results: a $1.061$-approximation algorithm for identical valuations (\Cref{thm:IdenticalVals}), and an exact algorithm for binary valuations (\Cref{thm:BinaryVals}). The proofs of these results are presented in \Cref{sec:Proof_identical,sec:Proof_binary} respectively.

\begin{restatable}[\textbf{Identical valuations}]{theorem}{IdenticalVals}
 \label{thm:IdenticalVals}
 Given any fair division instance with additive and identical valuations, there exists a polynomial time $1.061$-approximation algorithm for the Nash social welfare maximization problem.
\end{restatable}


\begin{restatable}[\textbf{Binary valuations}]{theorem}{BinaryVals}
 \label{thm:BinaryVals}
 Given any fair division instance with additive and binary valuations, a Nash optimal allocation can be computed in polynomial time.
\end{restatable}

\section{Identical Valuations: Proof of Theorem~\ref{thm:IdenticalVals}}
\label{sec:Proof_identical}
\input{NSWidentical.tex}

\section{Binary Valuations: Proof of Theorem~\ref{thm:BinaryVals}}
\label{sec:Proof_binary}
\input{NSWbinary.tex}

\bibliographystyle{alpha}
\bibliography{NSW}

\end{document}

%% file: NSWidentical.tex
This section provides the proof of \Cref{thm:IdenticalVals}, which we recall below:
\IdenticalVals*

Our proof of \Cref{thm:IdenticalVals} relies on two intermediate results: First, we will show in \Cref{lem:AlgIdentical_EFx} that the allocation computed by the greedy algorithm (called $\AlgId{}$, given in Algorithm~\ref{alg:IdenticalVals}) satisfies an approximate envy-freeness property called $\EFx{}$, defined below. We will then show in \Cref{lem:EFx_implies_approx_Nash} that any allocation with this property---in particular, the allocation computed by $\AlgId{}$---provides a $1.061$ approximation guarantee.

We start by describing the notion of envy-freeness and some of its variants.

\paragraph{Envy-freeness and its variants}
Given an instance $\langle [n], [m], \V \rangle$ and an allocation $A$, we say that an agent $i \in [n]$ \emph{envies} another agent $k \in [n]$ if $i$ prefers the bundle of $k$ over its own bundle, i.e., $v_i(A_k) > v_i(A_i)$. An allocation $A$ is said to be \emph{envy-free} ($\EF{}$) if each agent prefers its own bundle over that of any other agent, i.e., for every pair of agents $i,k \in [n]$, we have $v_i(A_i) \geq v_i(A_k)$. Likewise, an allocation $A$ is said to be \emph{envy-free up to the least positively valued good} ($\EFx{}$) if for every pair of agents $i,k \in [n]$, we have $v_i(A_i) \geq v_i(A_k \setminus \{j\})$ for every $j \in A_k$ such that $v_{i,j} > 0$. The notion of $\EFx{}$ first appeared in the work of Caragiannis et al.~\cite{CKM+16unreasonable}. Plaut and Roughgarden~\cite{PR18almost} study the existence of $\EFx$ allocations for special cases of the fair division problem.

We will now describe our algorithm called $\AlgId{}$.

\begin{algorithm}[t]
 \DontPrintSemicolon
 \KwIn{An instance $\langle [n], [m], \V \rangle$ with identical, additive valuations.}
 \KwOut{An allocation $A$.}
 \BlankLine
 Order the goods in descending order of value, i.e., $v(j_1) \geq v(j_2) \geq \dots v(j_m) >0$.\;
 Set $A \leftarrow (\emptyset,\emptyset,\dots,\emptyset)$.\;
 \For{$\ell = 1$ to $m$}{
 	Set $i \leftarrow \arg\min_{k \in [n]} v(A_k)$ \tcp*{ties are broken lexicographically}
 	$A_i \leftarrow A_i \cup \{j_\ell\}$ \tcp*{Allocate the good $j_\ell$ to the agent with the least valuation }
 }
 \KwRet $A$
 \caption{Greedy Algorithm for Identical Valuations (\AlgId{})}
\label{alg:IdenticalVals}
\end{algorithm}

\paragraph{Greedy algorithm for identical valuations}
As mentioned earlier in \Cref{sec:Introduction}, our algorithm $\AlgId{}$ (Algorithm~\ref{alg:IdenticalVals}) allocates the goods one by one in descending order of their value. In each iteration, a good is assigned to the agent with the least valuation. Assigning goods in this manner ensures that at each step, the algorithm picks the agent providing the greatest improvement in $\NSW$. It is easy to see that \AlgId{} runs in polynomial time. Our next result (\Cref{lem:AlgIdentical_EFx}) shows that $\AlgId{}$ always outputs an $\EFx{}$ allocation.

\begin{restatable}{lemma}{AlgIdEFx}
\label{lem:AlgIdentical_EFx}
The allocation $A$ returned by \AlgId{} is $\EFx{}$.
\end{restatable}
\begin{proof}
Let $A^\ell$ be the allocation maintained by \AlgId{} at the end of the $\ell^{\text{th}}$ iteration. It suffices to show that for each $\ell \in [m]$, if $A^{\ell - 1}$ is $\EFx{}$, then so is $A^\ell$.

Write $j_\ell$ to denote the good allocated in the $\ell^{\text{th}}$ iteration, and let $i$ be the agent that receives this good; thus $A^\ell_i = A^{\ell-1}_i \cup \{j_\ell\}$. Notice that only the valuation of agent $i$ is affected by the assignment of $j_\ell$, while the allocation any other agent $k \in [n] \setminus \{i\}$ is unchanged. Therefore, in order to establish that $A^\ell$ is $\EFx{}$, we only need to consider agent $i$ and show that $v(A^\ell_i \setminus \{j\}) \leq v(A^\ell_k)$ for all $k \in [n]$ and each $j \in A^{\ell}_i$. Since $\AlgId{}$ processes the goods in decreasing order of value, the good $j_\ell$ is the least valued good in $A^\ell_i$. Thus, for any $j \in A^\ell_i$, we have that $v(A^\ell_i \setminus \{ j \}) \leq v(A^\ell_i \setminus \{ j_\ell \}) = v (A^{\ell-1}_i) \leq v(A^\ell_k)$ for all $k \in [n]$; here, the last inequality follows from the agent selection rule of $\AlgId{}$, i.e., the fact that $i \in \arg\min_{k \in [n]} v(A^{\ell-1}_k)$. This shows that $A^\ell$ must be $\EFx{}$.
\end{proof}

Our final result in this section shows that any $\EFx{}$ allocation provides a $1.061$ approximation to Nash social welfare when the valuations are additive and identical.

\begin{restatable}{lemma}{EFxApproxNash}
\label{lem:EFx_implies_approx_Nash}
Let $\I = \langle [n], [m], \V \rangle$ be an instance with additive and identical valuations, and let $A$ be an $\EFx{}$ allocation for $\I$. Then, $\NW(A) \geq \frac{1}{1.061} \NW(A^*)$, where $A^*$ is the Nash optimal allocation for $\I$.
\end{restatable}
\begin{proof}
For notational convenience, we reindex the bundles in the allocation $A$ such that $v(A_1) \geq v(A_2) \geq \dots \geq v(A_n)$, where $v$ denotes the (additive and identical) valuation function for all agents. Let $\ell := \min_k v(A_k)$ denote the valuation of the least valued bundle under $A$ (thus $v(A_n) = \ell$).

For any agent $k \in [n-1]$ with two or more goods in $A_k$, $\EFx{}$ property implies that
\begin{equation}
v(A_k) \leq 2\ell.
\label{eqn:EFx_condition2}
\end{equation}
In particular, \Cref{eqn:EFx_condition2} implies that if $v(A_k) > 2\ell$ for some $k \in [n-1]$, then $A_k$ consists of exactly one good. Let $S := \{k \in [n] \, : \, v(A_k) > 2\ell\}$ denote the set of agents with such singleton bundles. Write $s = |S|$ and let $A_S := \{j_1, j_2, \ldots, j_s \}$ denote the set of goods owned by the agents in $S$.

For analysis, we will now consider a set of allocations where only the goods in $A_S$ are required to be allocated integrally, and any other good can be allocated fractionally among the agents. Formally, we define a \emph{partially-fractional allocation} $B \in [0,1]^{n \times m}$ as follows: for every good $j \in A_S$, $B_{i,j} \in \{0,1\}$ for any agent $i \in [n]$ subject to $\sum_i B_{i,j} = 1$, and for any other good $j \in [m] \setminus A_S$, $B_{i,j} \in [0,1]$ for any agent $i \in [n]$ subject to $\sum_i B_{i,j} = 1$. We let $\F$ denote the set of all such partially-fractional allocations, and let $A^\F$ denote the Nash optimal allocation in $\F$.\footnote{The valuation of an agent under a fractional allocation $B$ is given by $v(B_i) = \sum_j v(j)B_{i,j} $.} Since all integral allocations belong to $\F$, we have $\NW(A^\F) \geq \NW(A^*)$. Therefore, in order to prove the lemma, it suffices to show that $\NW(A) \geq \frac{1}{1.061} \NW(A^\F)$.

Define $\alpha := \min_{k \in [n]} v(A^\F_k) / \ell$. Observe that all goods in $A_S$, namely $j_1,j_2,\dots,j_s$, must belong to separate bundles in $A^\F$. This is because the combined value of all goods in $[m] \setminus A_S$ is strictly less than $2\ell(n-s)$. Therefore, if two (or more) goods in $A_S$ belong to the same bundle in $A^\F$ (say, $A^\F_a$), then there must exist another bundle in $A^\F$ (say, $A^\F_b$) with value strictly less than $2\ell$. In that case, we can simply swap the bundle $A^\F_b$ with one of the goods (of value more than $2 \ell$) in $A^\F_a$ and strictly improve $\NSW$, which is a contradiction. Therefore, without loss of generality, each good in $A_S$ belongs to a unique bundle in $A^\F$. Using this observation, we can reindex the bundles in $A^\F$ such that $j_i \in A^\F_i$ for all $i \in S$. 

It is easy to see that $\alpha \geq 1$.\footnote{If $\alpha < 1$, then it must be that $v(A^\F_n) < \ell$, i.e. a nonzero amount of fractional good is taken away from the bundle $A_n$. In that case, one can reassign (part of) this fractional good---currently assigned to one of the agents in $[n-1]$--- to $A^\F_n$ and strictly improve the Nash social welfare, contradicting the assumption that $A^\F$ is the Nash optimal in $\F$.} In addition, we can show that $\alpha < 2$. Indeed, as argued above, $\sum_{i > s} v(A^\F_i) < 2 \ell (n-s)$. This implies that $\alpha \ell = \min_i v(A^\F_i) < 2 \ell$, i.e., $\alpha <2$. Using this bound we can establish a useful structural property of $A^\F$: For all $i \in S$, the bundles $A^\F_i$ are singletons (i.e., $A^\F_i = \{ j_i \}$ for all $i \in S$) and, for all $k \notin S$, we have $v(A^\F_k) = \alpha \ell$. This follows from the observation that any bundle $A^\F_k$ in $A^\F$ which has a fractionally allocatable good (say, good $j$) is of value equal to $\alpha \ell = \min_a v(A^\F_a)$; otherwise, we can ``redistribute'' $j$ between $A^\F_k$ and $\arg\min_{k} v(A^\F_k)$ to obtain another fractional allocation with strictly greater $\NW$. Moreover, since for any $i \in S$, $j_i \in A^\F_i$ and $v(j_i) > \alpha \ell$ (recall that $\alpha < 2$), the bundle $A^\F_i$ does not contain a fractionally allocatable good. This, in particular, implies that $\cup_{i \in S}  A^\F_i = A_S$. All the remaining goods in $[m] \setminus A_S$ are fractionally allocatable, and hence the bundles $A^\F_k$ for all $k \notin S$ are of value equal to $\alpha \ell$. This structural property gives us the following bound for $\NW(A^\F)$:
\begin{align}
\NW(A^\F) & = \left(\prod_{i \in S} v(A^\F_i) \cdot \prod_{i \in [n] \setminus S} v(A^\F_i) \right)^{1/n} = \left( \prod_{i \in S} v(A_i) \cdot \left( \alpha \ell \right)^{(n-s)}\right)^{1/n}. 
\label{eqn:Nash_Upperbound}
\end{align}

We will now provide a lower bound for $\NW(A)$ that will allow us to prove the desired approximation guarantee. This is done by constructing an allocation $A' \in \F$ such that $\NW(A') \leq \NW(A)$. Along with \Cref{eqn:Nash_Upperbound}, this provides an analysis-friendly lower bound for the quantity $\NW(A') / \NW(A^\F)$. 

We start with the initialization $A' \leftarrow A$. Next, while there exist two agents $i,k \in [n]$ such that $\ell<v(A'_i)<v(A'_k)<2\ell$, we transfer goods of value $\Delta = \min\{ v(A'_i)-\ell,2\ell - v(A'_k) \}$ from $A'_i$ (the lesser valued bundle) to $A'_k$ (the larger valued bundle). Such a transfer is possible because the goods in the bundles with value less than $2\ell$ are allowed to be allocated fractionally. Notice that the Nash social welfare does not increase as a result of this transfer. Also, it is easy to see that this process terminates, since after each iteration of the while loop, either $v(A'_i)=\ell$ or $v(A'_k)=2\ell$ or both, and hence some agent can take no further part in any future iterations. Upon termination of the above procedure, there can be at most one agent (say, $r$) such that $v(A'_r)\in (\ell,2\ell)$; for every other agent $k\in [n]\setminus S$, $v(A'_k) \in \{\ell, 2\ell\}$.

Let $T = \{ k \in [n] \, : \, v(A'_k) \geq 2\ell \}$ and let $t = |T|$. Notice that by construction of $A'$, $S \subseteq T$; hence, $s \leq t$. We then have the following bound on the Nash social welfare of the allocation $A'$:
\begin{align}
\NW(A') & = \left( \prod_{i \in S} v(A'_i) \cdot \prod_{i \in T \setminus S} v(A'_i) \cdot  \prod_{i \in [n] \setminus T} v(A'_i) \right)^{1/n}  \geq \left( \prod_{i \in S} v(A_i) \cdot \left( 2 \ell \right)^{(t-s)} \cdot \ell^{(n-t)} \right)^{1/n}.
\label{eqn:Nash_Lowerbound}
\end{align} 

Let $\phi = \sum_{k \in [n] \setminus S} v(A_k)$ denote the combined value of all goods except for those in the set $A_S$. We will now use the allocations $A^\F$ and $A'$ to obtain upper and lower bounds for $\phi$, which in turn will help us achieve the desired approximation ratio for the allocation $A$. 

First, recall that the goods in the set $A_S = \{j_1, j_2, \ldots, j_s\}$ are allocated as singletons in $A^\F$ to the bundles $i \in S$ . Along with the fact that $v(A^\F_k) =  \alpha\ell$ for all $ k \in [n] \setminus S$, this gives
\begin{equation}
\phi = \sum_{k \in [n] \setminus S} v(A^\F_k) = (n-s) \cdot \alpha\ell.
	\label{eqn:phi_lowerbound}
\end{equation}

Next, in the allocation $A'$, each bundle corresponding to agents in $T \setminus S$ is valued at exactly $2\ell$, and that for each agent in $[n] \setminus T$ (except for the agent $r$) is valued at exactly $\ell$. By overestimating $v(A'_r)$ to be $2\ell$, we get
\begin{equation}
	\phi \leq 2\ell(t + 1 -s) + \ell (n-t - 1).
	\label{eqn:phi_upperbound}
\end{equation}

\Cref{eqn:phi_lowerbound,eqn:phi_upperbound} together imply that
\begin{equation}
\frac{t-s}{n-s} \geq \alpha - 1 - \frac{1}{n-s}.
\label{eqn:t_lowerbound}
\end{equation}

We can lower bound the quantity of interest $\frac{\NSW(A)}{\NSW(A^\F)}$, as below:
\begin{align*}
\frac{\NSW(A)}{\NSW(A^\F)} & \geq \frac{\NSW(A')}{\NSW(A^\F)}&\\
	& \geq \frac{\left( \prod_{i \in S} v(A_i) \cdot \left( 2 \ell \right)^{(t-s)} \cdot \ell^{(n-t)} \right)^{1/n}}{\left( \prod_{i \in S} v(A_i) \cdot \left( \alpha \ell \right)^{(n-s)}\right)^{1/n}} & (\text{from \Cref{eqn:Nash_Upperbound,eqn:Nash_Lowerbound}})\\
	& = \left( \frac{2^{t-s}}{\alpha^{n-s}} \right)^{1/n} &\\
	& \geq \left( \frac{2^{t-s}}{\alpha^{n-s}} \right)^{1/(n-s)} & \Bigg( \text{since } \left( \frac{2^{t-s}}{\alpha^{n-s}} \right)^{1/n} \leq \frac{\NSW(A)}{\NSW(A^\F)} \leq 1 \Rightarrow \frac{2^{t-s}}{\alpha^{n-s}} \leq 1 \Bigg)\\
	& \geq \frac{2^{\alpha - 1 }}{\alpha} & (\text{from \Cref{eqn:t_lowerbound} and for large }n)\\
	& \geq \frac{1}{2} e \ln 2 \approx \frac{1}{1.061} & (\text{minimum at } \alpha = 1/\ln 2 \approx 1.44).
\end{align*}
In the penultimate inequality, the reason for using the approximation $2^{\alpha - 1 - \frac{1}{n-s}} \approx 2^{\alpha - 1}$ is as follows: Imagine constructing a \emph{scaled-up} instance $\I'$ consisting of $c$ copies of the instance $\I$, where $c$ is arbitrarily large. Notice that $\I'$ has additive and identical valuations. Moreover, for any allocation $A$ that is a $\beta$-approximation for the instance $\I$, the allocation $B=(A,A,\dots,A)$ is $\beta$-approximation for $\I'$. Similarly, $A$ is $\EFx$ for $\I$ if and only if $B$ is $\EFx$ for $\I'$. Finally, write $n',s',\alpha',\ell'$ to denote the analogues of $n,s,\alpha,\ell$ in $\I'$. It is easy to see that $n' = cn$, $s' = cs$, $\alpha' = \alpha$ and $\ell' = \ell$. Since the agent with the least valuation in $\I$ (under allocation $A$) values his bundle at strictly below $2\ell$, we know that $s < n$, and thus the quantity $n' - s' = c(n-s)$ can be made arbitrarily large for appropriately chosen $c$. We can therefore ignore the term $\frac{1}{n-s}$ in the exponent of $2$ without loss of generality. This completes the proof of \Cref{lem:EFx_implies_approx_Nash}.
\end{proof}

The following example shows that the approximation guarantee of \Cref{lem:EFx_implies_approx_Nash} is almost tight.

\begin{example}[\textbf{Tightness of approximation factor for $\EFx{}$ allocations}]
\label{eg:Approx_factor_tight_for_EFx}
Consider a fair division instance with $m$ goods ($m$ is even) and $n = 2$ agents, where the (additive and identical) valuations are given as follows: $v(j_1) = v(j_2) = m - 2$, and $v(j_\ell) = 1$ for $\ell \in \{3,4,\dots,m\}$. Notice that the allocation $A = \{(j_1,j_2),(j_3,\dots,j_m)\}$ is $\EFx{}$. Additionally, $\NW(A) = \left( (2m-4) \cdot (m-2) \right)^{1/2}$. It is also clear that $\NW(A^*) = \frac{3}{2}(m-2)$. The approximation ratio of $A$ is given by 
$$\frac{\NW(A)}{\NW(A^*)} = \left( \frac{\left( 2(m-2) \right) \cdot (m-2)}{(3(m-2)/2) \cdot (3(m-2)/2)} \right)^{1/2} \approx \frac{1}{1.0607},$$
which closely matches the approximation guarantee of \Cref{lem:EFx_implies_approx_Nash}.
\end{example}

%% file: NSWbinary.tex
This section provides the proof of \Cref{thm:BinaryVals}, which we recall below:
\BinaryVals*

Our proof of \Cref{thm:BinaryVals} relies on a greedy algorithm (Algorithm~\ref{alg:BinaryVals}, hereafter referred to as $\AlgBinary{}$). Starting from any suboptimal allocation, $\AlgBinary{}$ identifies a pair of agents such that a chain of swaps between them provides the greatest improvement in Nash social welfare (from among all pairs of agents). \Cref{lem:Binary_main} quantifies the progress towards the Nash optimal allocation made by $\AlgBinary{}$ in each step. As it turns out, the algorithm is required to run for at most $2m(n+1) \ln (nm)$ iterations. Overall, this provides a polynomial time algorithm for computing a Nash optimal allocation for binary valuations. The detailed description of $\AlgBinary{}$ follows.

\paragraph{Greedy algorithm for binary valuations}
The input to $\AlgBinary{}$ is an instance with additive and binary valuations along with a suboptimal allocation, and output is a Nash optimal allocation. At each step, the algorithm performs a greedy local update over the current allocation. Specifically, given a partition $A = (A_1, A_2, \ldots, A_n)$, \AlgBinary{} constructs a directed graph $G(A)$ as follows: There is a vertex for each agent (hence $n$ vertices overall), and between any pair of vertices $u$ and $v$, there are $|\Gamma_v \cap A_u|$ parallel edges directed from $u$ to $v$.\footnote{Recall that $\Gamma_i := \{ j \in [m] \, : \, v_{i,j} > 0 \}$.} A directed edge $(u,v)$ exists if and only if there exists a good that is valued by $v$ and is currently assigned to $u$. Observe that a directed simple path $P=(u_1, u_2, \ldots, u_k)$ in $G(A)$ corresponds to a sequence of reallocations. For each directed edge $(u_i, u_{i+1})$, there exists a good $j \in A_{u_i}$ that can be reassigned to $u_{i+1}$ via the updates $A_{u_i} \leftarrow A_{u_{i}} \setminus \{ j \}$ and $A_{u_{i+1}} \leftarrow A_{u_{i+1}} \cup \{ j \}$.

\begin{algorithm}[t]
 \DontPrintSemicolon
 \KwIn{An instance $\langle [n], [m], \V \rangle$ with binary, additive valuations, and a partition $A$.}
 \KwOut{A Nash optimal partition $A'$.}
 \BlankLine
 Set $A^0 \leftarrow A$.\;
 \For{$i = 1$ to $2m(n+1) \ln (nm)$}{
 	Construct the graph $G(A^{i-1})$ for the current partition $A^{i-1}$.\;
 	Set $R \leftarrow \{ (u,v) \in [n] \times [n] \, : \, v \text{ is reachable from } u \text{ in } G(A^{i-1}) \} $.\;
 	\ForEach{$(u,v) \in R$}{
 		Set $A^{i-1}(u,v) \leftarrow$ The partition obtained by reallocating along some path from $u$ to $v$.\;
 	}
 	\uIf{$\max_{(u,v) \in R} \NW(A^{i-1}(u,v)) > \NW(A^{i-1})$}{ 	
	Update $A^{i} \leftarrow \argmax {A^{i-1}(u,v) \, : \, (u,v) \in R } \NW(A^{i-1}(u,v))$.
	}
	\Else{\KwRet $A^{i-1}$}
 }
 \caption{Greedy Algorithm for Binary Valuations (\AlgBinary{})}
\label{alg:BinaryVals}
\end{algorithm}

Let $A(P)$ denote the partition obtained by reallocating goods along the path $P=(u_1, u_2, \ldots, u_k)$. Such a reallocation increases (decreases) the valuation of $u_k$ ($u_1$) by one, while the valuations of all intermediate agents $u_2,\dots,u_{k-1}$ are unchanged. The algorithm \AlgBinary{} greedily selects a specific path $P$ in $G(A)$, and reallocates the goods along $P$ to obtain the partition $A':= A(P)$. \Cref{lem:Binary_main} below describes the progress towards the optimal solution made by such a reallocation.

\begin{lemma}
\label{lem:Binary_main}
Given a suboptimal partition $A$, there exist agents $u$ and $v$ such that $v$ is reachable from $u$ in $G(A)$, and reallocating along \textbf{any} directed path $P$ from $u$ to $v$ leads to a partition $A':= A(P)$ that satisfies
\begin{align*}
\ln\NW(A^*) - \ln\NW(A') \leq \left( 1 - \frac{1}{m} \right) &\left( \ln\NW(A^*) - \ln\NW(A) \right).
\end{align*}
Here $A^*$ denotes the Nash optimal partition.
\end{lemma}

\begin{remark}
Note that there can be multiple paths $P$ from $u$ to $v$ in $G(A)$, and different goods that can be reallocated along a fixed edge of $P$, which might lead to different partitions $A(P)$. However, the Nash social welfare of \emph{any} resulting partition is the same, since the valuation of $u$ ($v$) goes down (up) by one and that of every other agent remains the same. Hence, the choice of path between a fixed pair of vertices is inconsequential.
\label{rem:Pick_Any_Path}
\end{remark}

In the remainder of this section, we will show that \Cref{lem:Binary_main} can used to prove \Cref{thm:BinaryVals}, followed by a proof of \Cref{lem:Binary_main}.

\begin{proof}[Proof of \Cref{thm:BinaryVals}]
\Cref{lem:Binary_main} ensures that if there does not exist an improving reallocation, then the current allocation $A^{i-1}$ is optimal. Hence, for the rest of the proof, we will focus on the case wherein the for-loop executes for all $2m(n+1) \ln (nm)$ steps.

The update rule followed by \AlgBinary{} and \Cref{lem:Binary_main} together guarantee that at the end of iteration $i$, we have
\begin{align*}
\ln\NW(A^*) - \ln\NW(A^i) \leq \left( 1 - \frac{1}{m} \right)  & ( \ln\NW(A^*) - \ln\NW(A^{i-1}) ).
\end{align*}

Repeated use of the above bound gives
\begin{align*}
\ln\NW(A^*) - \ln\NW(A^{i}) \leq \left( 1 - \frac{1}{m} \right)^{i} & ( \ln\NW(A^*) - \ln\NW(A^{0}) ).
\end{align*}

Since \AlgBinary{} executes for $2m(n+1) \ln (mn)$ iterations, the difference between the optimal partition $A^*$ and the partition $A'$ returned by the algorithm is given by
\begin{align*}
\ln\NW(A^*) &- \ln\NW(A')  \\
& \leq \left( 1 - \frac{1}{m} \right)^{2m(n+1) \ln (nm)}  ( \ln\NW(A^*) - \ln\NW(A^{0}) ) \\
& \leq \frac{1}{e^{2(n+1)\ln (nm)}}  ( \ln\NW(A^*) - \ln\NW(A^{0}) ) \\
& \leq \frac{1}{(nm)^{2(n+1)}} \ln\NW(A^*) \\
& \leq \frac{\ln m}{(nm)^{2(n+1)}} \qquad \qquad \text{(since $\NW(A^*) \leq m$ for binary valuations)}\\
& \leq \frac{1}{n m^{2n}} \qquad \qquad \qquad \ \text{(since $\ln m \leq m$, and $n,m \geq 2$)}\\
& < \frac{1}{n} \ln \left( 1+ \frac{1}{m^n} \right) \qquad \text{(since $\ln(1+x) > x^2$ for $0 < x < 0.5$)}.
\end{align*}
Thus, $\prod_{i \in [n]} v_i(A^*_i) < \prod_{i \in [n]} v_i(A'_i) \left( 1+\frac{1}{m^n} \right)$. We already know that $\prod_{i \in [n]} v_i(A'_i) \leq \prod_{i \in [n]} v_i(A^*_i)$. Since the valuations are assumed to be integral, and $\prod_{i \in [n]} v_i(A'_i) \leq m^n$, we have that $\NW(A^*) = \NW(A')$. Hence, $A'$ is Nash optimal. This completes the proof of \Cref{thm:BinaryVals}.
\end{proof}

We will now provide a proof of \Cref{lem:Binary_main}.

\begin{proof}[Proof of \Cref{lem:Binary_main}]
Our proof of existence of the desired path $P$ in the graph $G(A)$ is made convenient by the formulation of another graph $G^*(A)$. This graph is utilised only in the analysis of the algorithm and never explicitly constructed.

Recall that $A^*$ refers to a Nash optimal allocation. Consider the directed graph $G^*(A)$ consisting of $n$ vertices, one for each agent, and a directed edge $(u,v)$ for each good $j \in A_u \cap A_v^*$. The edge $(u,v)$ indicates that the good $j$ must be transferred from $u$ to $v$ to reach the optimal partition $A^*$. Note that the total number of edges in $G^*(A)$ is at most $m$.

Besides defining the graph $G^*(A)$, we also classify the agents depending on their valuation relative to $A^*$. In particular, let $\mathcal{E}$ and $\mathcal{D}$ denote the set of agents with \emph{excess} and \emph{deficit} valuations respectively, i.e., $\mathcal{E}:=\{ u \in [n] \, : \, |A_u| > |A^*_u| \}$ and $\mathcal{D}:=\{ v \in [n] \, : \, |A_v| < |A^*_v| \}$.\footnote{For binary valuations and a non-wasteful allocation $A$, we have $v_i(A) = |A|$ for each agent $i \in [n]$.} Any agent $ t \in [n] \setminus (\mathcal{E} \cup \mathcal{D})$ satisfies $|A_t| = |A^*_t|$.

The remainder of the proof consists of two parts: First, we will show that the edge set of $G^*(A)$ can be partitioned into simple directed paths $\mathcal{P}=\{P_1, P_2, \ldots, P_k\}$ and cycles $\mathcal{C} =\{C_1, C_2, \ldots \}$ such that each path $P_i \in \mathcal{P}$ starts at a vertex in $\mathcal{E}$ and ends at a vertex in $\mathcal{D}$. Second, we will use this decomposition to argue that one of the paths $P_i \in \mathcal{P}$ leads to a partition $A':= A(P_i)$ that satisfies the bound in \Cref{lem:Binary_main}. The lemma will then follow by observing that the edges of $P_i$ are also contained in the graph $G(A)$ constructed by \AlgBinary{}. Note that the existence of $P_i$ shows that the end vertex of $P_i$ (say, $v$) is reachable from the start vertex of $P_i$ (say, $u$) in $G(A)$. As noted earlier in \Cref{rem:Pick_Any_Path}, reallocating along \emph{any} path between $u$ and $v$ leads to the stated improvement in Nash social welfare.

We will start by proving the claim about decomposition of the edge set of $G^*(A)$. Consider a graph $H^*$ where for each vertex $u$ of $G^*(A)$, we include $\max (\text{indegree}(u),\text{outdegree}(u))$ vertices, say $\{u^1, u^2, \ldots\}$. Suppose the vertex $u$ has $\ell$ incoming edges and $\ell'$ outgoing edges in $G^*(A)$. To construct $H^*$, first we pick an arbitrary one-to-one assignment between the incoming edges and $\{ u^1, u^2, \ldots, u^\ell\}$. Similarly, each outgoing edge gets uniquely assigned to one of the vertices in $\{u^1, u^2, \ldots, u^{\ell'}\}$. With these assignments in hand, for every directed edge $e=(u,v)$ in $G^*(A)$, we include a directed edge $(u^i, v^j)$ in $H^*$ if and only if $e$ is assigned to $u^i$ and $v^j$. It is easy to see that each edge in $H^*$ corresponds to an edge in $G^*(A)$ and vice versa. 

Notice that each vertex in $H^*$ has at most one incoming and at most one outgoing edge. Furthermore, if $u^i$ is a source in $H^*$, then $u \in \mathcal{E}$. Similarly, if $v^j$ is a sink in $H^*$, then $v \in \mathcal{D}$. These properties together imply that the edges in $H^*$ can be partitioned into paths and cycles such that each path starts at a vertex $u^i$ with $ u \in \mathcal{E}$ and ends at a vertex $v^j$ with $v \in \mathcal{D}$. The correspondence between the edges of $H^*$ and $G^*(A)$ gives us the desired collection of paths $\mathcal{P}=\{P_1, P_2, \ldots, P_k\}$ and cycles $\mathcal{C}$ in $G^*(A)$.\footnote{We can ensure that the paths in $\mathcal{P}$ are \emph{simple} by removing cycles from each $P_i$ and placing such cycles in $\mathcal{C}$.} The aforementioned properties also imply that the paths in $H^*$ are edge-disjoint, therefore $k \leq m$.

We will now show that for one of the paths $P_i \in \mathcal{P}$ in $G^*(A)$ (and therefore, also in $G(A)$), the partition $A' := A(P_i)$ achieves the bound in \Cref{lem:Binary_main}. First, observe that reallocating along a cycle in $G^*(A)$ does not change the Nash social welfare. Hence, in order to reach a Nash optimal partition starting from $A$, it suffices to reallocate goods along the paths $P_1, P_2, \ldots, P_k$. Moreover, since the paths in $\mathcal{P}$ are edge disjoint in the graph $\H^*$, they correspond to reallocation of disjoint sets of goods. This means that the reallocations corresponding to a path $P_i \in \mathcal{P}$ can be performed independently of those corresponding to another path $P_j \in \mathcal{P}$.

Next, consider the sequence of partitions $B^1, B^2, \ldots, B^k$, obtained by successively reallocating along the paths $P_1, P_2, \ldots, P_k$. That is, $B^1 = A(P_1)$, $B^2 =B^1(P_2)$ and so on. Thus, the partition $B^k$ must be Nash optimal, i.e., $\NW(A^*) =  \NW(B^k)$. Consider the telescoping sum given by $\ln\NW(A^*) - \ln\NW(A) = \sum_{i=1}^{k-1} \ln\NW(B^{i}) - \ln\NW(B^{i-1})$, where $B^0 = A$. Since $k \leq m$, there must exist $i \in [k]$ such that 
\begin{equation}
\label{ineq:int0}
\ln\NW(B^{i}) - \ln\NW(B^{i-1}) \geq \frac{1}{m} \left( \ln\NW(A^*) - \ln\NW(A) \right).
\end{equation}
We will now show that the partition $A' := A(P_i)$ satisfies 
\begin{equation}
\label{ineq:int}
\ln\NW(A') - \ln\NW(A) \geq \ln\NW(B^{i}) - \ln\NW(B^{i-1}).
\end{equation}
Indeed, recall that each path in $\mathcal{P}$ starts at a vertex in $\mathcal{E}$ and ends at a vertex in $\mathcal{D}$. Hence, as we proceed through reallocations corresponding to $P_1, \ldots, P_k$, the cardinality of the set of goods assigned to any agent $u' \in \mathcal{E}$ is non-increasing and that of $v' \in \mathcal{D}$ is non-decreasing. Therefore, if $u$ ($v$) is the start (end) vertex of $P_i$, then $k_u  \geq  k'_u$ and $k_v \leq k'_v$, where $k_u$, $k_v$, $k'_u$ and $k'_v$ are the number of goods assigned to $u$ and $v$ in partitions $A$ and $B^{i-1}$ respectively. Since $\ln\NW(B^i) - \ln\NW(B^{i-1}) = \ln (k'_u-1) + \ln (k'_v+1) - (\ln k'_u + \ln k'_v )$ and $\ln\NW(A') - \ln\NW(A) = \ln (k_u-1) + \ln (k_v+1) - (\ln k_u + \ln k_v )$, the concavity of $\ln(\cdot)$ implies \Cref{ineq:int}. Finally, \Cref{ineq:int,ineq:int0} give us the desired relation
$$\ln\NW(A^*) - \ln\NW(A') \leq \left( 1 - \frac{1}{m} \right) (\ln\NW(A^*) - \ln\NW(A)).\qedhere$$
\end{proof}

Notice that the proof of \Cref{lem:Binary_main} works exactly the same way when for each agent $i$, $v_i(A_i) = f_i(|A_i|)$ for some \emph{concave} function $f_i$. That is, the valuation of an agent can be an (agent-specific) concave function of the cardinality (i.e., the number of nonzero valued goods owned by the agent). Thus, $\AlgBinary{}$ can find a Nash optimal allocation in polynomial time even when the valuation functions of agents are concave in cardinality. This observation is formalized in \Cref{cor:BinaryValsConcaveUtility}.

\begin{restatable}{corollary}{BinaryValsConcaveUtility}
 \label{cor:BinaryValsConcaveUtility}
 Given any fair division instance with concave and binary valuations, a Nash optimal allocation can be computed in polynomial time.
\end{restatable}

\begin{remark}
\label{rem:BinaryVals_Concave_Budget_additive}
A well-studied class of valuation functions captured by \Cref{cor:BinaryValsConcaveUtility} is that of \emph{budget-additive} valuations~\cite{LLN01combinatorial}. Under this class, the valuation of an agent $i \in [n]$ for a set of goods $G \subseteq [m]$ is given by $v_i(G) := \min\{c_i,\sum_{j \in G} v_{i,j}\}$, where $c_i > 0$ is an (agent-specific) constant, known as the \emph{utility cap}.

Garg et al.~\cite{GHM18Approximating} recently gave a $(2.404+\varepsilon)$-approximation algorithm for this class (for any $\varepsilon > 0$). For binary valuations, a budget-additive valuation function turns out to be a special case of the concave-in-cardinality functions mentioned above. Hence, by \Cref{cor:BinaryValsConcaveUtility}, a Nash optimal allocation can be found in polynomial time when the valuations are binary and budget-additive. It is unclear whether the existing techniques for finding a Nash optimal allocation under binary and additive valuations~\cite{DS15maximizing} admit a similar generalization.
\end{remark}